\documentclass[english,conference,letter]{IEEEtran}
\usepackage[T1]{fontenc}
\usepackage[latin9]{inputenc}
\usepackage{amsthm}
\usepackage{amsmath}
\usepackage{amssymb}
\usepackage{esint}

\makeatletter
  \theoremstyle{definition}
  \newtheorem{defn}{\protect\definitionname}
  \theoremstyle{plain}
  \newtheorem{thm}{\protect\theoremname}
  \theoremstyle{plain}
  \newtheorem{lem}{\protect\lemmaname}
  \theoremstyle{plain}
  \newtheorem{cor}{\protect\corollaryname}

\IEEEoverridecommandlockouts

\usepackage{dsfont}\usepackage{babel}
\usepackage{relsize}\usepackage{subdepth}\allowdisplaybreaks

\author{
    \IEEEauthorblockN{Dionysios S. Kalogerias$^{\dagger}$ and Athina P. Petropulu\\}
    \IEEEauthorblockA{Department of Electrical and Computer   Engineering, Rutgers, The State University of New Jersey}
    \{d.kalogerias, athinap\}@rutgers.edu
    \thanks{$^{\dagger}$ Corresponding author.}
    \thanks{This work is supported by the National Science Foundation under Grant CNS-1239188.}
}

\usepackage{babel}
\providecommand{\definitionname}{Definition}
  \providecommand{\lemmaname}{Lemma}
\providecommand{\corollaryname}{Corollary}
\providecommand{\theoremname}{Theorem}

\makeatother

\usepackage{babel}
\providecommand{\corollaryname}{Corollary}
\providecommand{\definitionname}{Definition}
\providecommand{\lemmaname}{Lemma}
\providecommand{\theoremname}{Theorem}

\begin{document}

\title{On the Coherence Properties of\\
 Random Euclidean Distance Matrices}
\maketitle
\begin{abstract}
In the present paper we focus on the coherence properties of general
random Euclidean distance matrices, which are very closely related
to the respective matrix completion problem. This problem is of great
interest in several applications such as node localization in sensor
networks with limited connectivity. Our results can directly provide
the sufficient conditions under which an EDM can be successfully recovered
with high probability from a limited number of measurements.\end{abstract}

\section{Introduction}

Considering $N$ points (or nodes) lying in $\mathbb{R}^{d}$ with
respective positions $\mathbf{p}_{i}\in\mathbb{R}^{d},i\in\mathbb{N}_{N}^{+}\equiv\left\{ 1,2,\ldots,N\right\} $,
the $\left(i,j\right)$-th entry of a Euclidean distance matrix (EDM)
is defined as 
\begin{equation}
\boldsymbol{\Delta}\left(i,j\right)\triangleq\left\Vert \mathbf{p}_{i}-\mathbf{p}_{j}\right\Vert _{2}^{2}\in\mathbb{R}_{+},\;\forall\,\left(i,j\right)\in\mathbb{N}_{N}^{+}\times\mathbb{N}_{N}^{+}.
\end{equation}
EDMs appear in a large variety of engineering applications, such as
sensor network positioning and localization \cite{MontanariOhPositioning,Drineas2006,EkamRam2012}, distributed
beamforming problems that rely on second order statistics of the internode
channels \cite{LiPetropuluPoor2011}, or molecular conformation \cite{CrippenHavel1988},
where, using NMR spectroscopy techniques, the distances of the atoms
forming the protein molecule are estimated, leading
to the determination of its structure. Recently, significant attention
has been paid to the problem of EDM completion from partial distance
measurements in sensor networks, which corresponds to scenarios with limited
connectivity among the network nodes (see, e.g., \cite{MontanariOhPositioning}).
It can be shown \cite{MontanariOhPositioning} that the rank of
an EDM is less or equal than $d+2$ and hence, for a large number of
nodes, such a matrix is always of low-rank. Therefore, under certain conditions,
matrix completion can be used to recover the missing entries of the
matrix. In this paper, we focus on those properties of general random
EDMs, which can provide sufficient conditions, under which the EDM
completion problem can be successfully solved with high probability.

\textit{Relation to the literature - } The special case where the
node coordinates are drawn from $\mathcal{U}\left[-1,1\right]$ has
been studied in \cite{MontanariOhPositioning}, \cite{EkamRam2012},
\cite{OhThesis2010}. In our paper, in addition to studying the more general case,
we point out some issues in the respective proofs in \cite{EkamRam2012,OhThesis2010},
which result in incorrect results. The correct results can be obtained
as a special case of the general results that we propose. More details
on the above issues are provided in Section IV.

The paper is organized as follows. In Section II, we provide a
brief introduction to the problem of matrix completion, as presented
and analyzed in \cite{Candes&Recht2009}. In Section III, we present
our main results. Finally, in Section IV, we discuss the connection
of our results with the respective ones presented in \cite{MontanariOhPositioning,Drineas2006}
and \cite{EkamRam2012}.

\section{Low Rank Matrix Completion}

Consider a generic matrix $\mathbf{M}\in\mathbb{R}^{K\times L}$ of
rank $r$, whose compact Singular Value Decomposition (SVD) is given
by $\mathbf{M}=\mathbf{U}\boldsymbol{\Sigma}\mathbf{V}^{\boldsymbol{\mathit{T}}}\equiv\sum_{i\in\mathbb{N}_{r}^{+}}\sigma_{i}\left(\mathbf{M}\right)\mathbf{u}_{i}\mathbf{v}_{i}^{\boldsymbol{\mathit{t}}}$
and with column and row subspaces denoted as $U$ and $V$ respectively,
spanned by the sets $\left\{ \mathbf{u}_{i}\in\mathbb{R}^{K\times1}\right\} _{i\in\mathbb{N}_{r}^{+}}$
and $\left\{ \mathbf{v}_{i}\in\mathbb{R}^{L\times1}\right\} _{i\in\mathbb{N}_{r}^{+}}$,
respectively.

Let $\mathcal{P}\left(\mathbf{M}\right)\in\mathbb{R}^{K\times L}$
denote an entrywise sampling of $\mathbf{M}$. In all the analysis
that follows, we will adopt the theoretical framework presented in
\cite{Candes&Recht2009}, according to which one hopes to reconstruct
$\mathbf{M}$ from $\mathcal{P}\left(\mathbf{M}\right)$ by solving
the convex program 
\begin{equation}
\begin{array}{ll}
\mathrm{minimize} & \left\Vert \mathbf{X}\right\Vert _{*}\\
\mathrm{subject\, to} & \mathbf{X}\left(i,j\right)=\mathbf{M}\left(i,j\right),\quad\forall\,\left(i,j\right)\in\boldsymbol{\Omega},
\end{array}\label{eq:convex1}
\end{equation}
where the set $\boldsymbol{\Omega}$ contains all matrix coordinates
corresponding to the observed entries of $\mathbf{M}$ (contained
in $\mathcal{P}\left(\mathbf{M}\right)$) and where $\left\Vert \mathbf{X}\right\Vert _{*}$
represents the nuclear norm of $\mathbf{X}$.

Also in \cite{Candes&Recht2009}, the authors introduce the notion
of \textit{subspace coherence}, in order to derive specific conditions
under which the solution of (\ref{eq:convex1}) coincides with $\mathbf{M}$.
The formal definition of subspace coherence follows, in a slightly
more expanded form compared to the original definition stated in \cite{Candes&Recht2009}. 
\begin{defn}
\cite{Candes&Recht2009} Let $U\equiv\mathbb{R}^{r}\subseteq\mathbb{R}^{N}$
be a subspace spanned by the set of orthonormal vectors $\left\{ \mathbf{u}_{i}\in\mathbb{R}^{N\times1}\right\} _{i\in\mathbb{N}_{r}^{+}}$.
Also, define the matrix $\mathbf{U}\triangleq\left[\mathbf{u}_{1}\,\mathbf{u}_{2}\,\ldots\,\mathbf{u}_{r}\right]\in\mathbb{R}^{N\times r}$
and let $\mathbf{P}_{U}\triangleq\mathbf{U}\mathbf{U}^{\boldsymbol{\mathit{T}}}\in\mathbb{R}^{N\times N}$
be the orthogonal projection onto $U$. Then the coherence of $U$
with respect to the standard basis $\left\{ \mathbf{e}_{i}\right\} _{i\in\mathbb{N}_{N}^{+}}$
is defined as 
\begin{flalign}
\mu\left(U\right) & \triangleq\frac{N}{r}\sup_{i\in\mathbb{N}_{N}^{+}}\left\Vert \mathbf{P}_{U}\mathbf{e}_{i}\right\Vert _{2}^{2}\nonumber \\
 & \equiv\frac{N}{r}\sup_{i\in\mathbb{N}_{N}^{+}}\sum_{k\in\mathbb{N}_{r}^{+}}\left(\mathbf{U}\left(i,k\right)\right)^{2}.
\end{flalign}

\end{defn}
Additionally, the following crucial assumptions regarding the subspaces
$U$ and $V$ are of particular importance \cite{Candes&Recht2009}.\\
 \ \\
 $\mathbf{A0}\qquad\max\left\{ \mu\left(U\right),\mu\left(V\right)\right\} \leq\mu_{0}\in\mathbb{R}_{++}$.\\
 \ \\
 $\mathbf{A1}\qquad\left\Vert \sum_{i\in\mathbb{N}_{r}^{+}}\mathbf{u}_{i}\mathbf{v}_{i}^{\boldsymbol{\mathit{t}}}\right\Vert _{\infty}\leq\mu_{1}\sqrt{\dfrac{r}{KL}},\quad\mu_{1}\in\mathbb{R}_{++}$.\\
 \ \\
 Indeed, if the constants $\mu_{0}$ and $\mu_{1}$ associated with
the singular vectors of a matrix $\mathbf{M}$ are known, the following
theorem holds. 
\begin{thm}
\textup{\cite{Candes&Recht2009}} Let $\mathbf{M}\in\mathbb{R}^{K\times L}$
be a matrix of rank $r$ obeying $\mathbf{A0}$ and $\mathbf{A1}$
and set $N\triangleq\max\left\{ K,L\right\} $. Suppose we observe
$m$ entries of $\mathbf{M}$ with matrix coordinates sampled uniformly
at random. Then there exist constants $C$, $c$ such that if 
\begin{equation}
m\geq C\max\left\{ \mu_{1}^{2},\mu_{0}^{1/2}\mu_{1},\mu_{0}N^{1/4}\right\} Nr\beta\log N
\end{equation}
for some $\beta>2$, the minimizer to the program (\ref{eq:convex1})
is unique and equal to $\mathbf{M}$ with probability at least $1-cN^{-\beta}$.
For $r\leq\mu_{0}^{-1}N^{1/5}$ this estimate can be improved to 
\begin{equation}
m\geq C\mu_{0}N^{6/5}r\beta\log N,
\end{equation}
with the same probability of success. 
\end{thm}
Of course, the lower the rank of $\mathbf{M}$, the less the required
number of observations for achieving exact reconstruction. Regarding
the rank the EDM $\boldsymbol{\Delta}$ at hand, one can easily prove
the following lemma \cite{MontanariOhPositioning}. 
\begin{lem}
Let $\boldsymbol{\Delta}\in\mathbb{R}^{N\times N}$ be an EDM corresponding
to the distances $N$ of points (nodes) in $\mathbb{R}^{d}$. Then,
$rank\left(\boldsymbol{\Delta}\right)\leq d+2$. 
\end{lem}
Thus, in the most common cases of interest, that is, when $d$ equals
$2$ or $3$ and for sufficiently large number of nodes $N$, $rank\left(\boldsymbol{\Delta}\right)\ll N$
and consequently the problem of recovering
$\boldsymbol{\Delta}$ from a restricted number of observations is of great interest.

\section{The Coherence of Random EDMs}

According to Theorem 1, the bound of the minimum number of observations
required for the exact recovery of $\boldsymbol{\Delta}$ from $\mathcal{P}\left(\boldsymbol{\Delta}\right)$,
involves the parameters $\mu_{0}$ and $\mu_{1}$. Next, we derive
a general result, which provides estimates of these parameters in
a probabilistic fashion. 
\begin{thm}
Consider $N$ points (nodes) lying almost surely in a $d$-dimensional
convex polytope defined as $\mathcal{H}_{d}\triangleq\left[a,b\right]^{d}$
with $a<0<b$ and let $\mathbf{p}_{i}\triangleq\left[x_{i1}\, x_{i2}\,\ldots\, x_{id}\right]^{\boldsymbol{\mathit{t}}}\in\mathcal{H}_{d},i\in\mathbb{N}_{N}^{+}$
denote the position of node $i$, where $x_{ij}\in\mathcal{H}_{1},j\in\mathbb{N}_{d}^{+}$
represents its respective $j$-th position coordinate. Assume that
each $x_{ij},\left(i,j\right)\in\mathbb{N}_{N}^{+}\times\mathbb{N}_{d}^{+}$
is drawn independently and identically according to an arbitrary but
atomless probability measure $P$ with finite moments up to order
$4$. Define 
\begin{flalign*}
m_{k} & \triangleq\mathbb{E}\left\{ x_{ij}^{k}\right\} \equiv\int x^{k}dP<\infty,
\end{flalign*}
$k\in\mathbb{N}_{4}^{+},\forall\,\left(i,j\right)\in\mathbb{N}_{N}^{+}\times\mathbb{N}_{d}^{+}$
and let, without any loss of generality, $m_{1}\equiv0$. Also, define
$c\triangleq\max\left\{ \left|a\right|,\left|b\right|\right\} $ and
pick a $t\in\left[0,1\right]$ and a probability of failure $\gamma\in\left[0,1\right]$.
Then, as long as 
\begin{equation}
N\geq\dfrac{2\theta\left(\log\left(d+2\right)-\log\left(\gamma\right)\right)}{\left(1-t\right)^{2}},\label{eq:Condition_N}
\end{equation}
the associated Euclidean distance matrix $\boldsymbol{\Delta}$ with
worst case rank $d+2$ obeys the assumptions $\mathbf{A0}$ and $\mathbf{A1}$
with 
\begin{equation}
\mu_{0}=\dfrac{\theta}{t\left(d+2\right)}\quad and\quad\mu_{1}=\dfrac{\theta}{t\sqrt{d+2}},\label{eq:coherence_constants}
\end{equation}
with probability of success at least $1-\gamma$, where the constant
$\theta$ (that is, independent of $N$) is defined as 
\begin{equation}
\theta\triangleq\dfrac{1+dc^{2}+d^{2}c^{4}}{\lambda^{*}},
\end{equation}
with $\lambda^{*}\triangleq\min\left\{ \lambda_{1},\lambda_{2},\lambda_{3},m_{2}\right\} ,$
and where \textup{$\lambda_{1},\lambda_{2},\lambda_{3}$ }are the
real and positive solutions to the cubic equation\textup{ 
\begin{equation}
\sum_{i\in\mathbb{N}_{3}}\alpha_{i}\lambda^{i}=0,
\end{equation}
}with\textup{ 
\begin{flalign}
\alpha_{0} & \triangleq \left(m_{4}m_{2} - m_{2}^{3} - m_{3}^{2}\right)d,\label{eq:alpha0}\\
\alpha_{1} & \triangleq-m_{2}^{3}d^{2}+\nonumber \\
 & +\left(m_{2}^{3}+m_{3}^{2}+m_{2}^{2}-m_{4}-m_{4}m_{2}\right)d-m_{2},\label{eq:alpha1}\\
\alpha_{2} & \triangleq m_{2}^{2}d^{2}+\left(m_{4}-m_{2}^{2}\right)d+m_{2}+1\quad\mathrm{and}\label{eq:alpha2}\\
\alpha_{3} & \triangleq-1.\label{eq:alpha3}
\end{flalign}
} 
\end{thm}
Before we proceed with the proof of the theorem, let us state a well
known result from random matrix theory, the matrix Chernoff
bound (exponential form - see \cite{TroppTails2010}, Remark 5.3),
which will come in handy in the last part of the proof. 
\begin{lem}
\textup{\cite{TroppTails2010}} Consider a finite sequence of $N$
Hermitian and statistically independent random matrices $\left\{ \mathbf{F}_{i}\in\mathbb{C}^{K\times K}\right\} _{i\in\mathbb{N}_{N}^{+}}$
satisfying 
\begin{equation}
\mathbf{F}_{i}\succeq\mathbf{0}\quad and\quad\lambda_{max}\left(\mathbf{F}_{i}\right)\leq R,\quad\forall\, i\in\mathbb{N}_{N}^{+}
\end{equation}
and define the constants 
\begin{equation}
\xi_{min\left(max\right)}\triangleq\lambda_{min\left(max\right)}\left(\sum_{i\in\mathbb{N}_{N}^{+}}\mathbb{E}\left\{ \mathbf{F}_{i}\right\} \right).
\end{equation}
Define $\mathbf{F}_{s}\triangleq\sum_{i\in\mathbb{N}_{N}^{+}}\mathbf{F}_{i}$.
Then, it is true that 
\begin{multline}
\mathbb{P}\left\{ \lambda_{min}\left(\mathbf{F}_{s}\right)\leq t\xi_{min}\right\} \leq K\exp\left(-\dfrac{\xi_{min}}{2R}\left(1-t\right)^{2}\right),
\end{multline}
$\forall\, t\in\left[0,1\right]$, and 
\[
\mathbb{P}\left\{ \lambda_{max}\left(\mathbf{F}_{s}\right)\geq t\xi_{max}\right\} \leq K\left(\dfrac{e}{t}\right)^{\frac{t\xi_{max}}{R}},\quad\forall\, t\geq e.
\]

\end{lem}
\begin{proof}[Proof of Theorem 2]In order to make it more tractable,
we will divide the proof into the following subsections.

\subsection{Characterization of the SVD of $\boldsymbol{\Delta}$}

To begin with, observe that $\boldsymbol{\Delta}$ admits the rank-$1$
decomposition 
\begin{flalign}
\boldsymbol{\Delta} & =\boldsymbol{1}_{N\times1}\mathbf{p}^{\boldsymbol{\mathit{t}}}-2\sum_{i\in\mathbb{N}_{d}^{+}}\mathbf{x}_{i}\mathbf{x}_{i}^{\boldsymbol{\mathit{t}}}+\mathbf{p}\boldsymbol{1}_{N\times1}^{\boldsymbol{\mathit{t}}},
\end{flalign}
where 
\begin{flalign}
\mathbf{p} & =\left[\left\Vert \mathbf{p}_{1}\right\Vert _{2}^{2}\,\left\Vert \mathbf{p}_{2}\right\Vert _{2}^{2}\,\ldots\,\left\Vert \mathbf{p}_{N}\right\Vert _{2}^{2}\right]^{\boldsymbol{\mathit{t}}}\in\mathbb{R}_{+}^{N},\\
\mathbf{p}_{i} & \triangleq\left[x_{i1}\, x_{i2}\,\ldots\, x_{id}\right]^{\boldsymbol{\mathit{t}}}\in\mathcal{H}_{d},\quad i\in\mathbb{N}_{N}^{+}\quad\mathrm{and}\\
\mathbf{x}_{i} & \triangleq\left[x_{1i}\, x_{2i}\,\ldots\, x_{Ni}\right]^{\boldsymbol{\mathit{t}}}\in\mathcal{H}_{N},\quad i\in\mathbb{N}_{d}^{+}.
\end{flalign}
Then, we can equivalently express $\boldsymbol{\Delta}$ as 
\begin{equation}
\boldsymbol{\Delta}=\mathbf{X}\mathbf{D}\mathbf{X}^{\boldsymbol{\mathit{T}}},\label{eq:alter}
\end{equation}
where 
\begin{flalign}
\mathbf{X} & \triangleq col\left(\mathbf{X}_{1},\mathbf{X}_{2},\ldots,\mathbf{X}_{N}\right)\in\mathbb{R}^{N\times\left(d+2\right)},\\
\mathbf{X}_{i} & \triangleq\left[1\,\mathbf{p}_{i}^{\boldsymbol{\mathit{t}}}\,\left\Vert \mathbf{p}_{i}\right\Vert _{2}^{2}\right]\in\mathbb{R}^{1\times d+2},\quad i\in\mathbb{N}_{N}^{+}\quad\mathrm{and}\\
\mathbf{D} & \triangleq\mathrm{diag}\left(1,-2\boldsymbol{1}_{1\times d},1\right)\in\mathbb{R}^{\left(d+2\right)\times\left(d+2\right)}.
\end{flalign}
Since $\boldsymbol{\Delta}$ is a symmetric matrix, it can be easily
shown that its SVD possesses the special form 
\begin{flalign}
\boldsymbol{\Delta} & =\mathbf{U}\left|\boldsymbol{\Lambda}\right|\left(\mathbf{U}\cdot\mathrm{sign}\left(\boldsymbol{\Lambda}\right)\right)^{\boldsymbol{\mathit{T}}}\triangleq\mathbf{U}\boldsymbol{\Sigma}\mathbf{U}_{\pm}^{\boldsymbol{\mathit{T}}},\label{eq:SVD}\\
 & \equiv\sum_{i\in\mathbb{N}_{d+2}^{+}}\left|\lambda_{i}\left(\boldsymbol{\Delta}\right)\right|\mathbf{u}_{i}\left(\mathrm{sign}\left(\lambda_{i}\left(\boldsymbol{\Delta}\right)\right)\mathbf{u}_{i}^{\boldsymbol{\mathit{T}}}\right),
\end{flalign}
where $\left|\mathbf{A}\right|$ and $\mathrm{sign}\left(\mathbf{A}\right)$
denote the entrywise absolute value and sign operators on the matrix
$\mathbf{A}\in\mathbb{R}^{K\times L}$, respectively. In the expressions
above, $\boldsymbol{\Sigma}\equiv\left|\boldsymbol{\Lambda}\right|$,
$\boldsymbol{\Lambda}\in\mathbb{R}^{\left(d+2\right)\times\left(d+2\right)}$
is the diagonal matrix containing the (at most $d+2$) non zero eigenvalues
of $\boldsymbol{\Delta}$ in decreasing order of magnitude, denoted
as $\lambda_{i}\left(\boldsymbol{\Delta}\right),i\in\mathbb{N}_{d+2}^{+}$,
whose absolute values coincide with its singular values, that is,
$\sigma_{i}\left(\boldsymbol{\Delta}\right)\equiv\left|\lambda_{i}\left(\boldsymbol{\Delta}\right)\right|$,
and $\mathbf{U}\in\mathbb{R}^{N\times\left(d+2\right)}$ contains
as columns the eigenvectors of $\boldsymbol{\Delta}$ corresponding
to its non zero eigenvalues, denoted as $\mathbf{u}_{i},i\in\mathbb{N}_{d+2}^{+}$,
which essentially coincide with its left singular vectors.

Due to this special form of the SVD of $\boldsymbol{\Delta}$, if
we denote its column and row subspaces with $U$ and $U_{\pm}$, respectively,
it is true that 
\begin{flalign}
\mu\left(U_{\pm}\right) & =\frac{N}{d+2}\sup_{i\in\mathbb{N}_{N}^{+}}\sum_{k\in\mathbb{N}_{d+2}^{+}}\left(\mathrm{sign}\left(\lambda_{i}\left(\boldsymbol{\Delta}\right)\right)\mathbf{U}\left(i,k\right)\right)^{2}\nonumber \\
 & =\frac{N}{d+2}\sup_{i\in\mathbb{N}_{N}^{+}}\sum_{k\in\mathbb{N}_{d+2}^{+}}\left(\mathbf{U}\left(i,k\right)\right)^{2}\equiv\mu\left(U\right).
\end{flalign}
As a result, at least regarding the Assumption $\mathbf{A0}$, it
suffices to study the coherence of only one subspace, say $U$. It
is then natural to consider how the SVD of $\boldsymbol{\Delta}$,
given by (\ref{eq:SVD}), is related to its alternative representation
given by (\ref{eq:alter}).

Consider the thin QR decomposition of $\mathbf{X}$ given by $\mathbf{X}=\mathbf{V}\mathbf{A}$,
where $\mathbf{V}\in\mathbb{R}^{N\times\left(d+2\right)}$ with $\mathbf{V}^{\boldsymbol{\mathit{T}}}\mathbf{V}\equiv\mathbf{I}_{\left(d+2\right)}$
and $\mathbf{A}\in\mathbb{R}^{\left(d+2\right)\times\left(d+2\right)}$
constitutes an upper triangular matrix. Then, $\boldsymbol{\Delta}=\mathbf{V}\mathbf{A}\mathbf{D}\mathbf{A}^{\boldsymbol{\mathit{T}}}\mathbf{V}^{\boldsymbol{\mathit{T}}}$
and since the matrix $\mathbf{A}\mathbf{D}\mathbf{A}^{\boldsymbol{\mathit{T}}}\in\mathbb{R}^{\left(d+2\right)\times\left(d+2\right)}$
is symmetric by definition, the finite dimensional spectral theorem
implies that it is diagonalizable with eigendecomposition given by
$\mathbf{A}\mathbf{D}\mathbf{A}^{\boldsymbol{\mathit{T}}}=\mathbf{Q}\widetilde{\boldsymbol{\Lambda}}\mathbf{Q}^{\boldsymbol{\mathit{T}}}$,
where $\mathbf{Q}\in\mathbb{R}^{\left(d+2\right)\times\left(d+2\right)}$
with $\mathbf{Q}\mathbf{Q}^{\boldsymbol{\mathit{T}}}=\mathbf{Q}^{\boldsymbol{\mathit{T}}}\mathbf{Q}\equiv\mathbf{I}_{\left(d+2\right)}$
and $\widetilde{\boldsymbol{\Lambda}}\in\mathbb{R}^{\left(d+2\right)\times\left(d+2\right)}$
is diagonal, containing the eigenvalues of $\mathbf{A}\mathbf{D}\mathbf{A}^{\boldsymbol{\mathit{T}}}$.
Thus, we arrive at the expression 
\begin{equation}
\boldsymbol{\Delta}=\mathbf{V}\mathbf{Q}\widetilde{\boldsymbol{\Lambda}}\mathbf{Q}^{\boldsymbol{\mathit{T}}}\mathbf{V}^{\boldsymbol{\mathit{T}}}\equiv\mathbf{V}\mathbf{Q}|\widetilde{\boldsymbol{\Lambda}}|\left(\mathbf{V}\mathbf{Q}\cdot\mathrm{sign}\left(\widetilde{\boldsymbol{\Lambda}}\right)\right)^{\boldsymbol{\mathit{T}}},
\end{equation}
which constitutes a valid SVD of $\boldsymbol{\Delta}$, since $\left(\mathbf{V}\mathbf{Q}\right)^{\boldsymbol{\mathit{T}}}\mathbf{V}\mathbf{Q}\equiv\mathbf{I}_{\left(d+2\right)}$
and consequently, by the uniqueness of the singular values of a matrix,
$|\widetilde{\boldsymbol{\Lambda}}|\equiv\boldsymbol{\Sigma}$. Therefore,
we can set $\mathbf{U}=\mathbf{V}\mathbf{Q}$ and if $\mathbf{V}_{i}\in\mathbb{R}^{1\times\left(d+2\right)},i\in\mathbb{N}_{N}^{+}$
denotes the $i$-th row of $\mathbf{V}$, 
\begin{equation}
\mu\left(U\right)=\frac{N}{d+2}\sup_{i\in\mathbb{N}_{N}^{+}}\left\Vert \mathbf{V}_{i}\mathbf{Q}\right\Vert _{2}^{2}\equiv\frac{N}{d+2}\sup_{i\in\mathbb{N}_{N}^{+}}\left\Vert \mathbf{V}_{i}\right\Vert _{2}^{2}.
\end{equation}

\subsection{Bounding $\sup_{i\in\mathbb{N}_{N}^{+}}\left\Vert \mathbf{V}_{i}\right\Vert _{2}^{2}$}

Next, consider the hypotheses of the statement of Theorem 2. Since
the probability measure $P$ is atomless, the columns of $\mathbf{X}$
will be almost surely linearly independent. Then, $rank\left(\mathbf{A}\right)=d+2$
almost surely and consequently 
\begin{equation}
\left\Vert \mathbf{V}_{i}\right\Vert _{2}^{2}\leq\dfrac{\left\Vert \mathbf{X}_{i}\right\Vert _{2}^{2}}{\sigma_{min}^{2}\left(\mathbf{A}\right)},\quad\forall\, i\in\mathbb{N}_{N}^{+},
\end{equation}
where $\mathbf{X}_{i}\in\mathbb{R}^{1\times\left(d+2\right)},i\in\mathbb{N}_{N}^{+}$
denotes the $i$-th row of $\mathbf{X}$. Thus, in order to bound
$\mu\left(U\right)$ from above, it suffices to bound $\left\Vert \mathbf{X}_{i}\right\Vert _{2}^{2}$
from above and $\sigma_{min}^{2}\left(\mathbf{A}\right)$ from below.
Considering that $\mathbf{p}_{i}$ is bounded almost surely in $\mathcal{H}_{d}$,
we can easily bound $\left\Vert \mathbf{X}_{i}\right\Vert _{2}^{2}$
as 
\begin{equation}
\left\Vert \mathbf{X}_{i}\right\Vert _{2}^{2}\leq1+dc^{2}+d^{2}c^{4},\quad\forall\, i\in\mathbb{N}_{N}^{+}
\end{equation}
where $c\triangleq\max\left\{ \left|a\right|,\left|b\right|\right\} $.
Regarding $\sigma_{min}^{2}\left(\mathbf{A}\right)\equiv\lambda_{min}\left(\mathbf{A}^{\boldsymbol{\mathit{T}}}\mathbf{A}\right)$,
observe that the Grammian $\mathbf{A}^{\boldsymbol{\mathit{T}}}\mathbf{A}$
admits the rank-$1$ decomposition 
\begin{equation}
\mathbf{A}^{\boldsymbol{\mathit{T}}}\mathbf{A}\equiv\mathbf{A}^{\boldsymbol{\mathit{T}}}\mathbf{V}^{\boldsymbol{\mathit{T}}}\mathbf{V}\mathbf{A}=\mathbf{X}^{\boldsymbol{\mathit{T}}}\mathbf{X}=\sum_{i\in\mathbb{N}_{N}^{+}}\mathbf{X}_{i}^{\boldsymbol{\mathit{T}}}\mathbf{X}_{i}.\label{eq:Grammian}
\end{equation}
In general, it seems impossible to find a deterministic lower bound
for $\lambda_{min}\left(\mathbf{A}^{\boldsymbol{\mathit{T}}}\mathbf{A}\right)$.
For this reason, one could resort on probabilistic bounds that are
generally a lot easier to derive, providing considerably good estimates.
Towards this direction, below we will employ the matrix Chernoff bound
(Lemma 2), which fits perfectly to our bounding problem.

Observe that the sequence $\left\{ \mathbf{X}_{i}^{\boldsymbol{\mathit{T}}}\mathbf{X}_{i}\in\mathbb{R}^{N\times N}\right\} _{i\in\mathbb{N}_{N}^{+}}$
consists of $N$ statistically independent random Gramians, which
implies that $\mathbf{X}_{i}^{\boldsymbol{\mathit{T}}}\mathbf{X}_{i}\succeq\mathbf{0},\forall\, i\in\mathbb{N}_{N}^{+}$.
Also, since all these Gramians are rank-$1$, it can be trivially
shown that the one and only non zero eigenvalue of $\mathbf{X}_{i}^{\boldsymbol{\mathit{T}}}\mathbf{X}_{i}$
coincides with $\left\Vert \mathbf{X}_{i}\right\Vert _{2}^{2}$ and
consequently 
\begin{equation}
\lambda_{max}\left(\mathbf{X}_{i}^{\boldsymbol{\mathit{T}}}\mathbf{X}_{i}\right)\leq1+dc^{2}+d^{2}c^{4},\quad\forall\, i\in\mathbb{N}_{N}^{+}.
\end{equation}
As a result, the hypotheses of Lemma 2 are satisfied and our next
task involves specifying the constant $\xi_{min}$. Since the coordinates
$x_{ij},\left(i,j\right)\in\mathbb{N}_{N}^{+}\times\mathbb{N}_{d}^{+}$
are independent and identically distributed, it can be easily shown
that 
\begin{flalign}
\mathbb{E}\left\{ \mathbf{A}^{\boldsymbol{\mathit{T}}}\mathbf{A}\right\}  & =N\mathbb{E}\left\{ \mathbf{X}_{1}^{\boldsymbol{\mathit{T}}}\mathbf{X}_{1}\right\} \nonumber \\
 & =N\begin{bmatrix}1 & \mathbf{0}_{1\times d} & dm_{2}\\
\mathbf{0}_{d\times1} & m_{2}\mathbf{I}_{d} & m_{3}\mathbf{1}_{d\times1}\\
dm_{2} & m_{3}\mathbf{1}_{1\times d} & d\left(m_{4}-m_{2}^{2}\right)+d^{2}m_{2}^{2}
\end{bmatrix}\nonumber \\
 & \triangleq N\mathbf{R}_{d}.\label{eq:RRRRRR}
\end{flalign}
Thus, 
\begin{flalign*}
\xi_{min} & =\lambda_{min}\left(\mathbb{E}\left\{ \mathbf{A}^{\boldsymbol{\mathit{T}}}\mathbf{A}\right\} \right)=N\lambda_{min}\left(\mathbf{R}_{d}\right)
\end{flalign*}
and it suffices to characterize the eigenvalues of $\mathbf{R}_{d}\in\mathbb{R}^{\left(d+2\right)\times\left(d+2\right)}$.

\subsubsection{Positive Definiteness of $\mathbf{R}_{d}$}

We first argue that $\mathbf{R}_{d}$ is a positive definite matrix.
We can prove this argument using the strong Law of Large Numbers (LLN).
Since $\mathbf{A}$ is almost surely of full rank, the Gramian $\mathbf{A}^{\boldsymbol{\mathit{T}}}\mathbf{A}$
will be almost surely positive definite. Consider the sequence $\left\{ \mathbf{A}_{i}^{\boldsymbol{\mathit{T}}}\mathbf{A}_{i}\right\} _{i\in\mathbb{N}}$,
consisting of independent realizations of $\mathbf{A}^{\boldsymbol{\mathit{T}}}\mathbf{A}$.
Of course, $\mathbf{A}_{i}^{\boldsymbol{\mathit{T}}}\mathbf{A}_{i}\succ\mathbf{0},\forall\, i\in\mathbb{N}$.
Then, the strong LLN implies that 
\begin{equation}
\frac{1}{N}\sum_{i\in\mathbb{N}_{N}}\mathbf{A}_{i}^{\boldsymbol{\mathit{T}}}\mathbf{A}_{i}\overset{a.s.}{\longrightarrow}\mathbb{E}\left\{ \mathbf{A}^{\boldsymbol{\mathit{T}}}\mathbf{A}\right\} ,\;\mathrm{as}\; N\rightarrow\infty.
\end{equation}
As a consequence of the fact that the set of positive definite matrices
is closed under linear combinations with nonnegative weights, $\mathbb{E}\left\{ \mathbf{A}^{\boldsymbol{\mathit{T}}}\mathbf{A}\right\} \succ\mathbf{0}$
and, therefore, by (\ref{eq:RRRRRR}), $\mathbf{R}_{d}\succ\mathbf{0},\forall\, d\in\mathbb{N}^{+}$.

\subsubsection{Characterization of the eigenvalues of $\mathbf{R}_{d}$}

Next, in order to find the eigenvalues of $\mathbf{R}_{d}$, we would
like to solve the equation 
\begin{equation}
\det\left(\mathbf{R}_{d}-\lambda\mathbf{I}_{d+2}\right)=0.\label{eq:det}
\end{equation}
Using a well known identity, we can rewrite (\ref{eq:det}) as 
\begin{gather*}
\det\left(\mathbf{R}_{d}-\lambda\mathbf{I}_{d+2}\right)=\\
\det\left(\mathbf{E}-\lambda\mathbf{I}_{d+1}\right)\det\left((\mathbf{H}-\lambda)-\mathbf{G}\left(\mathbf{E}-\lambda\mathbf{I}_{d+1}\right)^{-1}\mathbf{F}\right)=\\
\left(1-\lambda\right)\left(m_{2}-\lambda\right)^{d}\cdot\\
\cdot\left(d\left(m_{4}-m_{2}^{2}\right)+d^{2}m_{2}^{2}-\lambda-\dfrac{d^{2}m_{2}^{2}}{1-\lambda}-\dfrac{dm_{3}^{2}}{m_{2}-\lambda}\right)=0,
\end{gather*}
where
\begin{flalign*}
\mathbf{E} & \triangleq\begin{bmatrix}1 & \mathbf{0}_{1\times d}\\
\mathbf{0}_{d\times1} & m_{2}\mathbf{I}_{d}
\end{bmatrix},\,\mathbf{F}\triangleq\begin{bmatrix}dm_{2}\\
m_{3}\mathbf{1}_{d\times1}
\end{bmatrix},\\
\mathbf{G} & \triangleq\begin{bmatrix}dm_{2} & m_{3}\mathbf{1}_{1\times d}\end{bmatrix},\,\mathbf{H}\triangleq d\left(m_{4}-m_{2}^{2}\right)+d^{2}m_{2}^{2},
\end{flalign*}
and which, after lots of dull algebra, yields the equation 
\begin{equation}
\left(m_{2}-\lambda\right)^{d-1}\left(\sum_{i\in\mathbb{N}_{3}}\alpha_{i}\lambda^{i}\right)=0,\label{eq:eigen}
\end{equation}
where the coefficients $\left\{ \alpha_{i}\right\} _{i\in\mathbb{N}_{3}}$
are given by (\ref{eq:alpha0}), (\ref{eq:alpha1}), (\ref{eq:alpha2})
and (\ref{eq:alpha3}), respectively.

Thus, $\mathbf{R}_{d}$ has an eigenvalue $\lambda_{0}\triangleq m_{2}$
with multiplicity $d-1$ for sure, and three additional eigenvalues
$\lambda_{1},\lambda_{2},\lambda_{3}$, which of course are the roots
to the cubic polynomial appearing in (\ref{eq:eigen}) and can be
computed easily in closed form. Since $\mathbf{R}_{d}\succ\mathbf{0},\forall\, d\in\mathbb{N}^{+}$,
all its eigenvalues must be strictly positive and, therefore, 
\begin{equation}
\xi_{min}=N\lambda_{min}\left(\mathbf{R}_{d}\right)=N\lambda^{*},
\end{equation}
where $\lambda^{*}\triangleq\min\left\{ \lambda_{1},\lambda_{2},\lambda_{3},m_{2}\right\} \in\mathbb{R}_{++}$.

\subsection{Putting it altogether}

We can now directly apply the matrix Chernoff bound (Lemma 2) to upper
bound the probability of the event $\mathcal{Z}\triangleq\left\{ \sigma_{min}^{2}\left(\mathbf{A}\right)\leq tN\lambda^{*}\right\} ,\forall\, t\in\left[0,1\right]$
as 
\begin{equation}
\mathbb{P}\left(\mathcal{Z}\right)\leq\left(d+2\right)\exp\left(-\dfrac{N\left(1-t\right)^{2}}{2\theta}\right)\triangleq\epsilon\left(t\right),\label{eq:inequality}
\end{equation}
where 
\begin{equation}
\theta\triangleq\dfrac{1+dc^{2}+d^{2}c^{4}}{\lambda^{*}}.
\end{equation}
The inequality (\ref{eq:inequality}) is equivalent to the statement
that $\sigma_{min}^{2}\left(\mathbf{A}\right)\geq tN\lambda^{*}$
with probability at least $1-\epsilon\left(t\right),\forall\, t\in\left[0,1\right].$
It is natural to select $N$ such that $1-\epsilon\left(t\right)\in\left[0,1\right]$.
Since the involved exponential function is strictly decreasing in
$N$, we can choose a $\gamma\in\left[0,1\right]$ such that $\epsilon\left(t\right)\leq\gamma$,
yielding the condition 
\begin{equation}
N\geq\dfrac{2\theta\left(\log\left(d+2\right)-\log\left(\gamma\right)\right)}{\left(1-t\right)^{2}}.\label{eq:condition}
\end{equation}
Consequently, as long as (\ref{eq:condition}) holds, the inequality
$\sigma_{min}^{2}\left(\mathbf{A}\right)\geq tN\lambda^{*}$ will
hold with probability at least $1-\gamma,$$\forall\,\gamma\in\left[0,1\right]$.

Hence, since $\left\Vert \mathbf{X}_{i}\right\Vert _{2}^{2}\leq1+dc^{2}+d^{2}c^{4},\forall\, i\in\mathbb{N}_{N}^{+}$
with probability $1$, we can write 
\begin{flalign}
\mu\left(U\right) & \leq\frac{N}{d+2}\sup_{i\in\mathbb{N}_{N}^{+}}\dfrac{\left\Vert \mathbf{X}_{i}\right\Vert _{2}^{2}}{\sigma_{min}^{2}\left(\mathbf{A}\right)}=\dfrac{1+dc^{2}+d^{2}c^{4}}{t\lambda^{*}\left(d+2\right)}\nonumber \\
 & \equiv\dfrac{\theta}{t\left(d+2\right)}\triangleq\mu_{0},\quad\forall\, t\in\left[0,1\right],
\end{flalign}
holding true with probability at least $1-\gamma,\forall\,\gamma\in\left[0,1\right]$
and with $N$ satisfying (\ref{eq:condition}).

Finally, regarding the Assumption $\mathbf{A1}$, by a simple argument
involving the Cauchy - Schwarz inequality, it can be shown that \cite{Candes&Recht2009}
\begin{equation}
\mu_{1}\triangleq\mu_{0}\sqrt{d+2}=\dfrac{\theta}{t\sqrt{d+2}},\quad\forall\, t\in\left[0,1\right],
\end{equation}
under the same of course circumstances as $\mu_{0}$, therefore completing
the proof.\end{proof}

If the coordinates of the nodes are drawn from a symmetric probability
distribution, the following more compact theorem holds. 
\begin{thm}
Consider $N$ points (nodes) lying almost surely in a $d$-dimensional
hypercube defined as $\mathcal{H}_{d}\triangleq\left[-a,a\right]^{d}$
with $a>0$ and let the rest of the hypotheses of Theorem 2 hold,
with the additional assumption $m_{3}\equiv0$. Pick a $t\in\left[0,1\right]$
and a probability of failure $\gamma\in\left[0,1\right]$. Then, as
long as the condition (\ref{eq:Condition_N}) holds, the associated
Euclidean distance matrix $\boldsymbol{\Delta}$ with worst case rank
$d+2$ obeys the assumptions $\mathbf{A0}$ and $\mathbf{A1}$ with
constants $\mu_{0}$ and $\mu_{1}$ defined as in (\ref{eq:coherence_constants}),
respectively, with probability of success at least $1-\gamma$ and
\begin{equation}
\theta\triangleq\dfrac{2\left(1+da^{2}+d^{2}a^{4}\right)}{\lambda^{*}},
\end{equation}
where 
\begin{flalign}
\lambda^{*} & \triangleq\min\left\{ \zeta-\sqrt{\zeta^{2}-4d\left(m_{4}-m_{2}^{2}\right)},2m_{2}\right\} 
\end{flalign}
and $\zeta\triangleq d\left(m_{4}-m_{2}^{2}\right)+d^{2}m_{2}^{2}+1$. 
\end{thm}
The proof of Theorem 3 is almost identical to that of Theorem 2. Therefore,
it is omitted. Further, if the node coordinates are drawn independently
from $\mathcal{U}\left[-1,1\right]$, the following corollary constitutes
a direct consequence of Theorem 2 and is also presented without proof. 
\begin{cor}
Consider $N$ points (nodes) lying almost surely in the $d$-dimensional
hypercube defined as $\mathcal{H}_{d}\triangleq\left[-1,1\right]^{d}$
and let $\mathbf{p}_{i}\triangleq\left[x_{i1}\, x_{i2}\,\ldots\, x_{id}\right]^{\boldsymbol{\mathit{t}}}\in\mathcal{H}_{d},i\in\mathbb{N}_{N}^{+}$
denote the position of node $i$, where $x_{ij}\in\mathcal{H}_{1},j\in\mathbb{N}_{d}^{+}$
represents its respective $j$-th position coordinate. Assume that
each $x_{ij},\left(i,j\right)\in\mathbb{N}_{N}^{+}\times\mathbb{N}_{d}^{+}$
is drawn independently and identically according to the uniform in
$\left[-1,1\right]$ probability measure, $\mathcal{U}\left[-1,1\right]$.
Pick a $t\in\left[0,1\right]$ and a probability of failure $\gamma\in\left[0,1\right]$.
Then, as long as the condition (\ref{eq:Condition_N}) holds, the
associated Euclidean distance matrix $\boldsymbol{\Delta}$ with worst
case rank $d+2$ obeys the assumptions $\mathbf{A0}$ and $\mathbf{A1}$
with constants $\mu_{0}$ and $\mu_{1}$ defined as in (\ref{eq:coherence_constants}),
respectively, with probability of success at least $1-\gamma$ and
\begin{equation}
\theta\triangleq\dfrac{90\left(1+d+d^{2}\right)}{\zeta-\sqrt{\zeta^{2}-720d}},
\end{equation}
where $\zeta\triangleq5d^{2}+4d+45$. 
\end{cor}

\section{Discussion Regarding Related Results}

It is important to note that for the very special case where the node
coordinates are drawn from $\mathcal{U}\left[-1,1\right]$ (see Corollary
1), to the best of the authors' knowledge, there exist two previously
published closely related results. In particular, see Theorem
1 in \cite{MontanariOhPositioning} and Remark 4.1.1 and the respective
proof in \cite{OhThesis2010}, as well as Lemma 2 and the respective
proof in \cite{EkamRam2012}.

However, there are some subtle technical issues in the aforementioned results. More
specifically, in both proofs, the respective authors claim that the
minimum eigenvalue of the matrix 
\begin{equation}
\begin{bmatrix}1 & \mathbf{0}_{1\times d} & \dfrac{d}{3}\\
\mathbf{0}_{d\times1} & \dfrac{1}{3}\mathbf{I}_{d} & \mathbf{0}_{d\times1}\\
\dfrac{d}{3} & \mathbf{0}_{1\times d} & \left(\dfrac{d}{3}\right)^{2}+\dfrac{4d}{45}
\end{bmatrix},
\end{equation}
which constitutes a special case of the matrix $\mathbf{R}_{d}$ defined
in (\ref{eq:RRRRRR}), equals $1/3$. One can easily confirm that this is not true.


Additionally, there is a minor oversight in both proofs, where it is implicitly
stated  that
since a Euclidean distance matrix is symmetric, its (real) left singular
vectors coincide with its right ones. This argument would be
correct only if one could prove that all EDMs belong to the
cone of symmetric and positive semidefinite matrices.
Further, such an argument is also incorrect since one can easily prove by
counterexample that there is at least one EDM that is not
positive semidefinite, i.e., with spectrum containing at least
one negative eigenvalue.
In the case of simply symmetric
matrices, the SVD takes the special form of (\ref{eq:SVD}). However,
we should note here that this mistake by itself does not affect the
outcome of the proofs, since the coherence of a matrix essentially
depends on the norm of the rows of the matrices whose columns constitute
the sets of its left and right singular vectors, respectively.

\section{Conclusion}

To the best of the authors' knowledge, Corollary 1 provides a novel result regarding the coherence of EDMs, for the special
case where the node coordinates are independently drawn from $\mathcal{U}\left[-1,1\right]$.
Furthermore, our main result, Theorem 1 presented above, provides
a substantial generalization to Corollary 1, essentially covering
any case where the node coordinates are independently drawn according
to an arbitrary, non - singular probability measure. Therefore, the
theoretical results presented in this paper can provide strong evidence
as well as sufficient conditions under which the EDM completion problem
can be successfully solved with high probability, a fact that also
justifies their direct practical applicability in various modern applications
which involve EDMs, such as sensor network localization and estimation
of the second order statistics of channel state information in wireless
communication networks.

 \bibliographystyle{IEEEtran}
\bibliography{IEEEabrv}

\end{document}